\documentclass[10pt,conference,compsoc]{IEEEtran}
\usepackage[utf8]{inputenc}
\usepackage[utf8]{inputenc}
\usepackage[T1]{fontenc}
\def\ps@headings{%
\def\@oddhead{\mbox{}\scriptsize\rightmark \hfil \thepage}%
\def\@evenhead{\scriptsize\thepage \hfil \leftmark\mbox{}}%
\def\@oddfoot{}%
\def\@evenfoot{}}
\makeatother
\usepackage{commath,amsmath}
\usepackage{amssymb}
\usepackage{amsthm}
\usepackage{dsfont}
\allowdisplaybreaks
\usepackage{algorithm}
\usepackage[noend]{algpseudocode}
\usepackage{float}
\usepackage{graphicx}
\usepackage{mathtools}
\usepackage[thinc]{esdiff}
\usepackage{epsf}
\usepackage{epsfig, amsmath, amssymb, graphicx}
\usepackage{subcaption}
\usepackage{multirow}
\usepackage[english]{babel}
\DeclareGraphicsExtensions{.pdf, .bmp, .eps, .ps, .jpeg, .png}

\newtheorem{theorem}{Theorem}

\long\def\comment#1{}

\makeatletter
\def\BState{\State\hskip-\ALG@thistlm}
\makeatother

\begin{document}

\title{TEE-based Selective Testing of Local Workers in Federated Learning Systems}

\author{Wensheng~Zhang, Trent~Muhr\\
Computer Science Department, Iowa State University\\
Ames, Iowa, 50011}%

\IEEEtitleabstractindextext{
\begin{abstract}
This paper studies ...
\end{abstract}

}

\maketitle

\IEEEpeerreviewmaketitle

\begin{abstract}
    This paper considers a federated learning system composed of
    a central coordinating server and 
    multiple distributed local workers,
    all having access to trusted execution environments (TEEs).
    In order to ensure that the untrusted workers
    correctly perform local learning,
    we propose a new TEE-based approach that also combines techniques from applied cryptography, smart contract and game theory.
    Theoretical analysis and implementation-based evaluations show that,
    the proposed approach is secure, efficient and practical. 
\end{abstract}

\section{Introduction}

Recent developments such as the advent of IoT \cite{Gubbi2013InternetOT} and 
an increasingly cloud-oriented healthcare industry \cite{Esposito2018BlockchainAP} have led to 
federated learning becoming an area of growing interest.
Federated learning~\cite{McMahan2017CommunicationEfficientLO} 
is a distributed, collaborative machine learning paradigm, 
where multiple edge devices or servers which hold local data samples 
collaboratively train their data to obtain a global model; 
a certain aggregation {\em server} is usually involved to coordinate such collaboration.
Here, the edge devices or servers are often called local {\em workers} of the system.  
As only model parameters, instead of data samples, are exchanged,
this decentralized machine learning paradigm is appealing 
due to its strength in protecting data privacy for the participants.  

Along with the popularity of federated learning comes a host of challenges, 
including communication efficiency, 
resilience to non-I.I.D. distribution of data samples, 
tolerance of dynamic participation, 
and providing security and privacy in the process.
Particularly, the distributed nature of federated learning 
introduces new security concerns.
It may be possible for a curious server to 
infer information about the data used by local workers 
during the training process. 
Making use of cryptographic primitives such as masking and public key cryptography,
secure aggregation~\cite{Bonawitz2017PracticalSA,Fereidooni2021SAFELearnSA,Xu2020VerifyNetSA, Phong2018PrivacyPreservingDL} 
has been an attempt to prevent this. 

Though there is a lot of research addressing the security/privacy risks due to 
possible misbehavior of aggregation server, 
it is imperative to also secure the system against misbehaving local workers.
Without proper security measure in place,
a local worker may deviate from the supposed honest behavior in various ways.
For instance, it may use faked rather than truthful data in local training;
not select its local training data as randomly as expected;
not use as many as expected data samples;
not honestly execute the local training operations.
Such misbehavior could be treated as poisonous attacks.
However, existing countermeasures to such attacks may not be sufficiently accurate or timely,
and not be able to identify and thus avoid/punish the misbehaving local workers.
When certain privacy protection mechanisms are applied at the server,
the inputs from different local workers could be "blindly" aggregated
which makes it even more challenging to 
detecting misbehavior and identify misbehaving local workers.

In this paper, we propose a different scheme
that detects misbehavior directly and immediately at the local workers in a trusted manner. 
The approach is developed based on the following ideas.
First, 
contemporary computers (e.g., servers, personal computers and mobile devices)
have been commonly equipped with TEEs
based on technologies such as Intel SGX~\cite{IntelSGXexplained} and TrustZone~\cite{trustzone}. 
In order to directly monitor the behavior of local workers,
we propose to deploy monitoring functions in these TEEs.  

Second, the monitoring function is implemented by having 
the TEEs to directly and immediately repeat a selected subset of 
the operations the local workers are expected to have conducted. 
However, 
the execution in TEEs is 
less efficient than in traditional untrusted environments. 
For example, in a computer with Intel SGX, 
the trusted memory space 
is much smaller than the regular memory space, which restricts
the performance that the TEEs can attain. 
To address this limitation,
we propose a game theoretic design 
to minimize the involvement of TEEs in the monitoring. 
We have proved that, 
by requiring
each economically-greedy local worker 
to make a deposit of a small amount 
(e.g., the cost for executing only one stage of training a neural network) 
when it joins the federated learning system, 
testing the correctness of only a small number (e.g., two) of 
the operations the worker is expected to conduct
can enforce it to behaves honestly. 

Third, 
directly testing a small number of operations,
such as the forward or backward propagation over a convolutional or fully-connected layer,
can still be very inefficient,
because even a single operation could involve large inputs. 
To address this issue,
we further propose to convert the heavy tests into lightweight tests.
This way, the TEE-based selective testing becomes more efficient and practical.

We implement our proposed TEE-based selective testing scheme for an Intel SGX-based computer,
and evaluate its performance for forward/backward propagation through convolutional/fully-connected layers 
during the training of a neural network model.
We also compare our scheme to two reference schemes:
the original scheme which conducts training in untrusted execution environment without any security measure;
the all-SGX scheme which conducts training completely in an SGX enclave thus ensures honest execution. 
The performance is measured by the running time
of our scheme in the SGX enclave and in the untrusted environment, 
as well as the running time of the original scheme and the all-SGX scheme.

As shown by the evaluation, 
our scheme only incurs a very testing cost in the TEE. 
For efficient selective-testing, 
our scheme, however, 
introduces extra operations (such as constructing Merkle hash trees) 
to be conducted in the untrusted execution environment,
which incurs the major overhead of the scheme.
The evaluation results indicate that
such overhead is comparable to the costs of the original and all-SGX schemes,
and it gets relatively smaller as the input/output scale increases.
For instance, for a convolutional layer
with $256\times 256$ inputs, 16 $8\times8$ filters and stride $2$,
our scheme spends in TEE only   
360 $\mu$s for forward and 658 $\mu$s for backward propagation,
and the time it spends in untrusted environment is 
84815 $\mu$s for forward and 173460 $\mu$s for backward propagation.
In comparison, the all-SGX scheme spends
76868 $\mu$s and 223482 $\mu$s for forward and backward propagation, respectively;
the original scheme spends 
74196 $\mu$s and 112449 $\mu$s for forward and backward propagation, respectively. 
Note that, the total time that our scheme spends is similar to that by the all-SGX scheme,
but the majority of our scheme's time is spent in the untrusted environment,
which is more easily to be reduced through parallelism. 
This is different for the all-SGX scheme, for which 
the execution time is all spent in the SGX enclave and thus is more difficult to reduce. 
Hence, our scheme is more feasible and efficient in practice.

In the rest of the paper,
Section 2 introduces background and problem description.
Section 3 describes and analyzes the basic framework of our proposed scheme.
Section 4 presents the enhancements that further improve the efficiency of our proposed scheme. 
Implementation-based evaluations are presented in Section 5.
Section 6 briefly reviews related work.
Finally, Section 7 concludes the paper.

\section{Problem Description}

\subsubsection*{System Model}

We consider a federated learning system composed of 
one central server and multiple distributed local workers.
Each local worker has both trusted execution environment (TEE) and 
untrusted execution environment. 
The computational and storage capacities of the TEE 
are much smaller than those of the untrusted environment.

Each local worker has its own training data,
which should never be exposed to others.
Coordinated by the central server,
the local workers collaborate in building a global neural network model.
We assume  
all the parties agree on the hyperparameters of the model, 
including the number of layers, 
the number of neurons on each layer, 
the connectivity between neurons, 
the activation functions used, etc. 
The central server has an initial model; 
then, the system works round by round to update it. 
In the beginning of each round, 
each local worker downloads the current global model 
from the central server,
and uses its own data to update the weights of connections. 
The central server collects the updates, and
applies them to the global model to get 
a newer version used in the next round. 

\subsubsection*{Model for Neural Network Training}

The model has ${\mathcal L}$ layers:
input layer $1$, 
hidden layers $2$, $\cdots$, ${\mathcal L}-1$, 
and final layer ${\mathcal L}$ that computes loss function and gradient.
Each layer $l\in[{\mathcal L}]$ has $n_l$ neurons.
For each hidden layer $l$, 
forward and backward propagation 
are conducted in two stages: transformation and activation. 
For the transformation stage, 
the forward propagation transforms the outputs of layer $l-1$ to the inputs of layer $l$,
while the backward propagation transforms the gradients of inputs to the gradients of outputs.
There are various transformation functions; 
we consider only full-connections and convolutions.

{\noindent\underline{Full Connection}}:
    Let $\Theta$ denote the weight matrix of the connections 
    from the outputs of layer $l-1$ to the inputs of layer $l$;
    specifically, $\Theta$ has $n_{l-1}$ rows and $n_{l}$ columns, 
    and each element $\theta_{i,j}$ on row $j$ and column $i$ 
    is the weight of the connection 
    from output $j$ of layer $l-1$ to input $i$ of layer $l$.
    Further let $\overrightarrow{X}$ and 
    $\overrightarrow{Y}$ denote 
    the vector of outputs from layer $l-1$ and the vector of inputs to layer $l$ respectively;
    let $\overrightarrow{\nabla X}$ and $\overrightarrow{\nabla Y}$ denote
    the gradients of the two vectors respectively.
    Then, the forward propagation conducts the transformation
    \begin{equation}
        \overrightarrow{Y} = (\Theta)^\intercal\times\overrightarrow{X}, 
    \end{equation}
    and the backward propagation conducts the transformation
    \begin{equation}
        \overrightarrow{\nabla X} = \Theta\times\overrightarrow{\nabla Y}.    
    \end{equation}
    Meanwhile, the backward propagation also computes the update for each $\theta_{i,j}$,
    denoted as $\nabla\theta_{i,j}$, as follows: 
    \begin{equation}
        \nabla\theta_{i,j} = -\eta\cdot\overrightarrow{\nabla Y}[i]\cdot\overrightarrow{X}[j], 
    \end{equation}
    where each $\vec{v}[i]$ represents the $i$-th element of vector $\vec{v}$ and $\eta$ is learning rate.

{\noindent\underline{Convolution}}: 
    Let $\overrightarrow{F^{(1)}}, \cdots, \overrightarrow{F^{(n_F)}}$ denote 
    the set of $n_F$ filters where
    each $\overrightarrow{F^{(t)}}$ has $\alpha_F\times\alpha_F$ elements denoted as
    $\overrightarrow{F^{(t)}}[i,j]$ for $i,j\in [\alpha_F]$,
    and $\delta$ denotes the stride.
    The output matrix $\vec{X}$ from layer $l-1$,
    which has $\alpha_X\times\alpha_X$ elements,
    can be viewed as the union of a two-dimensional array of grids.
    The array has $\alpha_Y=\lfloor 1+\frac{\alpha_X-\alpha_F}{\delta} \rfloor$ rows
    where each row also has $\alpha_Y$ columns, and
    each grid has $\alpha_F\times\alpha_F$ elements;
    every two consecutive grids on the same row (or column) have their starting points separated by $\delta$ elements. 
    
    During the forward propagation, 
    each filter $\overrightarrow{F^{(t)}}$ maps every grid in $\vec{X}$ to an element in a filtered image denoted as $\overrightarrow{Y^{(t)}}$,
    which is a matrix of $\alpha_Y\times\alpha_Y$ elements.
    Specifically, letting the elements of the grid on row $r$ and column $c$ be denoted as
    $\vec{X}[(r-1)\delta +i,(c-1)\delta+j]$ for $i,j\in [\alpha_F]$,
    then the element of $\overrightarrow{Y^{(t)}}$ on row $r$ and column $c$ is
    \begin{equation}\label{eq:conv-forward}
    \overrightarrow{Y^{(t)}}[r,c]=
    \sum_{i,j\in[\delta]} 
    \vec{X}[(r-1)\delta+i,(c-1)\delta+j]\cdot \overrightarrow{F^{(t)}}[i,j].
    \end{equation}
    In the rest of the paper, let $\vec{Y}$ denote 
    $\langle \overrightarrow{Y^{(1)}},\cdots,\overrightarrow{Y^{(n_F)}} \rangle$.
    
    During the backward propagation, 
    the gradients for $\vec{X}$ (denoted as $\overrightarrow{\nabla X}$)
    should be computed based on $\overrightarrow{Y^{(t)}}$ for $t\in [n_F]$ (denoted as $\overrightarrow{\nabla Y^{(t)}}$) and all the filters;
    meanwhile, the updates to the filters should also be computed. 
    Specifically, every element $\overrightarrow{\nabla X}[i,j]$ is computed as
    \begin{equation}\label{eq:conv-backward}
        \overrightarrow{\nabla X}[i,j]=\sum_{t\in[n_F]}\overrightarrow{\nabla X^{(t)}}[i,j],
    \end{equation}
    where for each $t$, 
    $\overrightarrow{\nabla X^{(t)}}[i,j]$ is computed as
    \begin{equation}\label{eq:conv-backward-1}
        \sum_{\phi(u,v,i,j)}
        \{\overrightarrow{\nabla Y^{(t)}}[u,v]\cdot\overrightarrow{F^{(t)}}[i-(u-1)\delta,j-(v-1)\delta]\},
    \end{equation}
    where $\phi(u,v,i,j)$ is defined as 
    \[
    (u,v \in [\alpha_Y]) \wedge \{i-(u-1)\delta,j-(v-1)\delta\in [\alpha_F]\}.
    \]
    Also, for each $t$ and every $i,j\in [\delta]$, $\overrightarrow{{\nabla F}^{(t)}}[i,j]$ is
    \begin{equation}\label{eq:conv-update}
        -\eta\sum_{u,v\in [\alpha_Y]} \overrightarrow{\nabla Y^{(t)}}[u,v]\cdot\overrightarrow{X}[(u-1)\delta+i,(v-1)\delta+j].
    \end{equation}

For the activation stage, 
the forward propagation feeds 
each input to an activation function,
denoted as $a(\cdot)$,
to get the corresponding output of layer $l$;
the backward propagation computes the gradients for the inputs
based on given gradients for the outputs and the definition of $a(\cdot)$.

\subsubsection*{Assumptions for Local Training Data}

Each local worker has its own training data, 
which is represented as records.
We assume that the validity of each record can be verified 
based on a digital signature mechanism. 
For example, 
it is reasonable to assume that valid medical data records should be  
digitally signed by certain authorized personnel and 
the digital signatures can be verified using
certain certified public keys,
so that any user knowing the certified public keys can 
verify such signatures and thus 
trust the information carried by the signed records. 
Hence, each record is assumed to bear the following format:
\begin{equation}
    \langle 
        x_1, \cdots, x_{n_X}; y_1, \cdots, y_{n_Y}; \sigma
    \rangle. 
\end{equation}
Here, 
$(x_1,\cdots,x_{n_X})$ is the vector of $n_X$ input features, 
$(y_1,\cdots,y_{n_Y})$ is the tag vector of $n_Y$ elements,
and $\sigma$ is a digital signature.
Particularly, 
a Merkle hash tree for the record is built 
with the hashes of $x_{i,1}, \cdots, x_{i,n_X}$ and 
$y_{i,1},\cdots,y_{i,n_Y}$ as leaf nodes,
the root of the above hash tree is called {\em record hash}, and
the hash is signed with an authorized private key to obtain a verifiable digital signature.

\subsubsection*{Security Assumptions and Goals}

In this work, 
we aim to address the following attacks 
that may be launched by a misbehaving local worker:
    using invalid (e.g., faked or modified) data for local learning;
    failing to choose training data randomly, 
    which is required by the commonly-used SGD method;
    failing to honestly conduct computation.
We assume local workers could be selfish or lazy, 
by pretending to have more data for training than they actually have, 
or by faking (skipping the complete procedure of) computation to save cost. 
Hence, we model them as economically-greedy; that is, 
they always intend to maximize their profits, computed as 
the incomes minus the costs that they have to pay.  

We do not consider the attack launched by the central server,
which may attempt to reveal the confidentiality of data owned by local workers. 
Such attacks can be addressed as follows:
each local worker reports encrypted updates, then
the TEE at the central server aggregates and decrypts the updates
to obtain a new global model. 

A variety of side-channel attacks have been discovered for SGX-based designs,
which are out of the scope of this paper. Note that,
our proposed scheme executes testing only after the untrusted local worker has
completed the tested tasks and submitted commitments which cannot be changed; 
hence, even if the worker can observe the execution of an enclave, 
it is not able to change its computation that has already been committed.

\comment{
Suppose the current model is denoted as
\begin{equation}
{\cal M} =
\langle
\vec{M}_0, \vec{C}_{0,1}, \vec{M}_1, \vec{C}_{1,2}, \cdots, \vec{C}_{L-2,L-1}, \vec{M}_{L-1}
\rangle,
\end{equation}
where each $\vec{M}_l$ is a matrix containing
the weights for all edges in layer $l$ and
$\vec{C}_{l,l+1}$ is the set of functions representing
the conversion from each output of layer $l$ to the corresponding input of layer $l+1$.
Particularly, $\vec{M}_l$ has $I_l$ rows and $O_l$ columns if the layer has $I_l$ inputs and $O_l$ outputs.

We assume each UL applies each model on $D$ inputs,
computes the cost functions,
aggregates the costs,
encrypts the result,
and reports the encrypted result to the TA securely.
Once the TA receives the encrypted results from all local learners,
it aggregates them together to get the final cost.
The final cost is securely transmitted to all ULs,
each then updating its model (i.e., the weights of edges) based on the final cost.
}

\section{The TEE-based Selective Testing Scheme}

\subsection{Primitives: Commitment and Verification}

\begin{algorithm}[htb]\small
\caption{Primitive for commitment and verification}\label{alg:commitment-basic}

{\bf Construct\_Commit($\vec{v}$)} 
\begin{algorithmic}[1]
\State construct a Merkle hash tree $MT(\vec{v})$ with $hash(\vec{v}[i])$ for $i=1,\cdots,|v|$ as leaf nodes;
\State let $comm(\vec{v})$ denote the root of $KT(\vec{v})$;
\State $\forall~i\in\{1,\cdots,|\vec{v}|\}$, let 
    \begin{equation}
        evid(\vec{v},i)=\langle h_1, h_2, \cdots, h_m\rangle
    \end{equation}
denote the sequence of co-path hashes for $hash(\vec{v}[i])$ on $MT(\vec{v})$, i.e., 
    \begin{equation}\scriptsize
        hash(\cdots hash(hash(\vec{v}[i]), h_1)\cdots,h_m)~=~comm(\vec{v}),
    \end{equation}
where $hash(x,y)$ is $hash(x|y)$ if $x$ is left sibling of $y$ on $MT(\vec{v})$ 
or $hash(y|x)$ otherwise;
\State return $\{comm(\vec{v}),~evid(\vec{v},1),~\cdots,~evid(\vec{v},|\vec{v}|)\}$.
\end{algorithmic}

{\bf Verify\_Element($u$, $i$, $comm(\vec{v})$, $evid$)} 
\begin{algorithmic}[1]
\State let $evid = \{h_1,\cdots,h_m\}$;
\State compute $c=hash(\cdots hash(hash(u), h_1)\cdots,h_m)$, where 
$hash(u)$ is treated as the $i$-th leaf node of $MT(\vec{v})$ and 
as in Construct\_Commit($\vec{v}$),
$hash(x,y)$ is $hash(x|y)$ if $x$ is left sibling of $y$ on $MT(\vec{v})$ 
or $hash(y|x)$ otherwise;
\If {$c=comm(\vec{v})$} 
\State return $True$;
\Else 
\State return $False$.
\EndIf
\end{algorithmic}
\end{algorithm}

We first introduce primitives
{\em Construct\_Commit} and {\em Verify\_Element}.
As formally presented in Algorithm~\ref{alg:commitment-basic},
primitive {\em Construct\_Commit} takes a vector $\vec{v}$ as input,
constructs a Merkle hash tree $MT(\vec{v})$ with the hashes of each elements of $\vec{v}$ as leaf nodes.
Then, the root of the tree is returned as the {\em commitment} of the vector.
Meanwhile, for each element $\vec{v}[i]$,
the sequence of its corresponding co-path hash values on $MT(\vec{v})$
is returned as the evidence for verifying it as the $i$-th element of $\vec{v}$.
Accordingly, primitive {\em Verify\_Element} takes four arguments,
i.e., an element $u$, an index $i$, the commitment $comm(\vec{v})$ for certain $\vec{v}$ and an evidence $evid$.
It assumes $u$ as the $i$-th element of $\vec{v}$ and 
makes use of the assumed co-path hash values in $evid$
to recompute the root of $MT(\vec{v})$.
If and only if the recomputed root is the same as $comm(\vec{v})$,
the element $u$ is confirmed.

\subsection{SIMD Computation}

When training a neural network model,
on each stage of each layer,
the same type of operation needs to be performed over different data.
Taking the forward propagation over a convolutional layer $l$ as example,
there are two stages.
For the first stage (transformation), 
the outputs of layer $l-1$ are transformed to the inputs of layer $l$ as follows:
the input of each neuron at layer $l$ is computed as 
the inner product of a filter matrix and a set of output elements of layer $l-1$.
For the second stage (activation),
at each neuron of layer $l$,
the input is fed to an activation function to obtain the output of the neuron. 
We call such computation paradigm at each stage of each layer as 
single instruction multiple data ({\em SIMD}) computation, 
and formalize it as $\vec{Y} = g(\vec{X})$,
where $g(\cdot)$ represents the operation,
$\vec{X}$ the vector of input and 
$\vec{Y}$ the vector of corresponding output.
When instantiated for the aforementioned transformation stage of convolutional layer $l$,
$g$ stands for the inner product operation,
$\vec{Y}$ is the vector of input to layer $l$,
and $\vec{X}$ is the vector containing all the subsets of 
layer $l-1$'s output elements used to compute the elements in $\vec{Y}$.
Therefore, the whole model training procedure can be formalized as  
a sequence of SIMD computations.

\subsection{The Proposed Selective Testing Scheme}

\begin{figure}[htp]
    \centering
    \includegraphics[width=8cm]{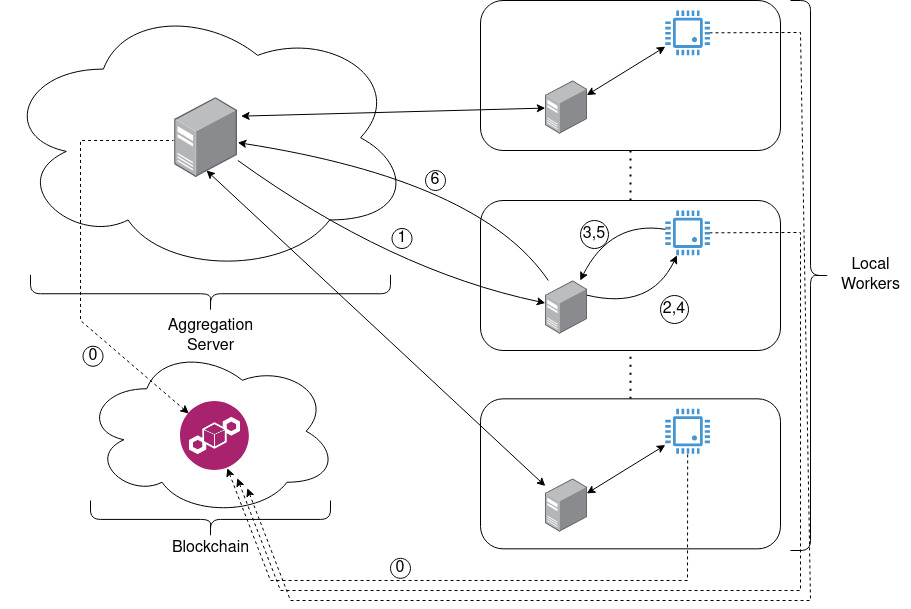}
    \caption{System overview: The system is comprised of an aggregation server and some number of untrusted workers, all with access to a blockchain with smart contract capabilities.
    (0) Each party participates in a smart contract.
    (1) The endorsed model is downloaded from the server.
    (2) Worker has the local enclave validate records for training and asks the enclave to select one for training.
    (3) TLM sends record choice to untrusted worker.
    (4,5) The untrusted worker and its local enclave engage in selective testing.
    (6) The final update is sent to the aggregation server, endorsed by the worker's enclave. 
    }
    \label{fig:overview}
\end{figure}

To effectively and efficiently verify if a local worker 
who participates federated learning has 
honestly conducted the procedure of training,
we propose a selective testing scheme 
that combines the techniques of 
game theory, applied cryptography and smart contracts on blockchain.

\subsubsection{System Components}

We define the following system components:
a central server (CS), 
multiple untrusted local workers (UW),
and one trusted local monitor (TLM) co-residing with each UW.
Here,
each TLM is run in a TEE.
When a UW joins the system,
the TLM co-located with the UW should authenticate itself to the CS. %
Then, the TLM should set up secret pairwise keys 
with the UW and the CS, respectively, 
to secure their communications. 
Also, we assume the CS does not collude with any UW.

\subsubsection{Signing Smart Contract}

The CS signs a smart contract with each UW.
With the contract, 
the UW makes a {\em small} deposit $d$ that is only required to be larger than 
twice of the maximal cost of executing one stage of SIMD computations.
If the UW is found dishonest by its co-located TLM through selective testing,
its deposit will be taken by the CS and it will be evicted from the system; 
otherwise, the UW will remain in the system and continue its participation.

\subsubsection{Validating and Preparing Local Data Records}

After a UW has signed the above smart contract with the CS,
it requests its co-located TLM to validate its data records 
and prepare them for federated learning. 
Each record 
$\langle$ $x_1$, $\cdots$, $x_{n_X}$, $y_1$, $\cdots$, $y_{n_Y}$, $\sigma$ $\rangle$
is processed as follows. 
First, the TLM checks the validity of the record.
That is, letting $\vec{v} = (x_1,\cdots, x_{n_X}, y_1, \cdots, y_{n_Y})$,
it computes $comm(\vec{v})=Construct\_Commit(\vec{v})$ and verifies if
$\sigma$ is a valid signature of $comm(\vec{v})$.
Second, 
the TLM assigns a unique identity $i\in [n_R]$ to the record, 
where $n_R$ is the number of such records. 
Thus, each record can be denoted as 
\begin{equation}
R_i = \langle i, x_{i,1}, \cdots, x_{i,n_X}, y_i, \cdots, y_{i,n_Y}, \sigma_i \rangle. 
\end{equation}
Then, a Merkle tree for all of the $n_R$ records is built with
$h_i=hash(i|\sigma_i)$ for all $i\in [n_R]$ as leaf nodes. 
The root hash of the tree is denoted as $h_R$.
The UW keeps this Merkle tree for later use, but 
the TLM only keeps $h_R$ and $n_R$.

\subsubsection{Initializing Each Round (i.e., testing for layer~1)}

After its local data records have been validated and prepared 
for federated learning by its co-located TLM,
a UW can formally participate the federated learning round by round.

The UW downloads the current global neural network model from the CS.
The model and its components should be signed by the CS so that
a malicious UW cannot modify them before they are given to the TLM.

For the simplicity of presentation, 
we assume that only one record is processed in each round
though our scheme can be extended for more general cases. 
For the purpose of randomly selecting data record for training 
(which is required by federated learning),
the TLM randomly selects an ID $i\in [n_R]$ at the beginning of a round, 
and asks the UW to pick the record with the ID for training.

In response, 
the UW retrieves the content of the selected record (i.e., $R_i$), 
the record hash (i.e., $h_i=hash(i|\sigma_i)$), and 
the corresponding co-path hash values on the Merkle tree of all $n_R$ records.
Then, it communicates $h_i$ and the corresponding co-path hash values,
which are called the {\em evidence} of the input, to the TLM.

Upon receiving the commitment,
the TLM verifies it by recomputing the root hash using $h_i$ and the evidence,
and checking if the recomputed root hash is the same as $h_R$. 
Once the verification succeeds, 
the TLM records $h_i$ and proceeds with the rest of the round;
otherwise, it identifies the UW as dishonest and quits the system.

\subsubsection{Testing for Each Hidden Layer}
\label{sec:select-test-general-model}

The operations at each hidden layer include one or more stages.
Along with a UW's execution at each stage,
its co-locating TLM conducts selective testing for the stage.
The operations of the UW and TLM, 
as well as their interactions, 
can be generally modelled as follows:

Suppose the SIMD computation at a stage 
has $n$ same-type computations.
Let $\vec{X}$ denote the input vector,
$\vec{Y}$ the output vector, and 
$g(.)$ the computation function.

Before this stage starts,
the TLM should have already obtained 
$comm(\vec{X})$ (i.e., the commitment for the input)
and the UW should be able to provide evidence for verifying each input.
Note that, if this stage is the first stage, 
the afore-described procedure for 
initializing each round has provided the detail on 
how the above are accomplished;
if this stage is not the first stage, as to be shown later,
its inputs should be the outputs of the previous stage,
for which the commitment and evidences should have been produced during the previous stage. 

This stage starts with the UW's execution.
The UW evaluates $g(.)$ with every element of $\vec{X}$ to obtain the corresponding output element in $\vec{Y}$.
Then, it computes the commitment and evidences for $\vec{Y}$ by calling $Construct\_Commit(\vec{Y})$,
keeps the results, and sends $comm(\vec{Y})$ to the TLM. 

Upon receiving $comm(\vec{Y})$,
the TLM randomly selects $p$ out of the $n$ computations to test.
For each of the selected computation $i\in [n]$,
with $\vec{X}[i]$ denoting the input element that should be used in the computation and 
$\vec{Y}[i]$ denoting the expected output element,
the testing is as follows:
The TLM requests the UW for 
the input element (denoted as $u_0$) and output element (denoted as $u_1$) of computation $i$,
as well as the evidences ($evid_0$ and $evid_1$ respectively) for 
verifying these elements to be $\vec{X}[i]$ and $\vec{Y}[i]$ respectively.  
Once receiving the above,
the TLM calls 
$Verify\_Element($ $u_0$, $i$, $comm(\vec{X})$, $evid_0)$ and
$Verify\_Element($ $u_1$, $i$, $comm(\vec{Y})$, $evid_1)$
to verify if $u_0=\vec{X}[i]$ and $u_1=\vec{Y}[i]$.
Then, it checks if $g(u_0)=u_1$.
If any of the tests fails, 
the TLM identifies UW as dishonest and stops participation.

\subsubsection{Testing for Layer $\cal L$}

The final layer computes the loss function during the forward propagation,
and computes the gradients for its input elements (i.e., the output elements from the last hidden layer).
Since these computations are not heavy,
the TLM directly repeat them.

\subsubsection{Endorsing Model Updates}

A TLM should endorse the model updates computed by its co-located UW 
as long as the UW is not found dishonest.
The CS only accepts a UW's model updates that have been endorsed by its co-located TLM;
a UW that fails to provide endorsed model updates 
is not allowed to get the current global model from the CS
and thus is evicted from the federated learning system. 

In our scheme, 
during the course of backward propagation,
the TLM tests the model updates made by the UW;
if the test succeeds, it signs the updates to endorse,
and the signature can be verified by the CS.

\comment{
the output of the layer $l-1$, which has $n_{l-1}$ elements,
should be multiplied with matrix $(\Theta^{(l)})^\intercal$
to obtain the $n_l$-element input of layer $l$.
Here, each of the $n_l$ input elements of layer $l$
is computed as  

In the course of training a neural network model 
or using the model for prediction, 
for each layer of the model, 
the same type of operation may need to be performed 
over a set of different elements.
To efficiently verify if these operations have been performed honestly,
we can apply a game theoretic approach to conduct selective testing. 
}

\comment{
The limitation of the basic version is that, 
for multiplication operation, 
the selective testing has the complexity of $O(n)$ where $n$ is the size of input or output. 
Hence, we further enhance the basic version by
introducing a more efficient selective testing algorithm for multiplication operations.
}

\subsection{Game-theoretic Analysis of Selective Testing}

We model the interactions between the CS and each UW as 
an infinite extensive game with perfect information, denoted as $G=(P,A,U)$. 
Here, $P=\{CS.TLM,UW\}$ is the set of players
where $CS.TLM$ represents the coalition including CS and TLM. 
$A$ is the set of actions taken by the players, including
all the combinations of the $n$ same-type computations to fake and
all the combinations of the $n$ computations to test.
As we treat the $n$ computations equally, 
the action set that the UW can take is denoted as $A_{uw}=\{0,1,\cdots,n\}$ 
where each element represents the number of computations 
that the UW randomly chooses to fake;
the action set that the CS can take is denoted as $A_{cs.tlm}=\{0,1,\cdots,n\}$
where each element represents the number of computations
that the CS.TLM randomly chooses to test.
$U=\{U_{uw},U_{cs.tlm}\}$ is the players' utility functions. 

The UW's utility %
is defined as:
\begin{eqnarray}
&&U_{uw}(A_{uw},A_{cs.tlm}) \\\nonumber
&=&    \begin{cases}
        B - (c_c(n) - c_c(A_{uw})) &\text{if not detected;} \\
        -d - (c_c(n) - c_c(A_{uw})) &\text{if detected.}
    \end{cases}
\end{eqnarray}
It says that, 
if none of the $A_{uw}$ faked computations is detected,
the UW's utility is $B - (c_c(n) - c_c(A_{uw}))$,
where $B$ is the UW's benefit from sharing the results of federated learning 
(by staying in the system) and
$c_c(x)$ is the cost of honestly executing all the $x$ computations.
Note that, here we assume that faking a computation does not have computation cost,
thus the computation cost is $c_c(n) - c_c(A_{uw})$ when
$A_{uw}$ of the $n$ computations are faked. 
If any of the $A_{uw}$ faked computations is detected,
the UW loses its deposit; hence, its utility becomes
$-d - (c_c(n) - c_c(A_{uw}))$.

Similarly, the CS.TLM's utility is defined as:
\begin{eqnarray}
&& U_{cs.tlm}(A_{uw},A_{cs.tlm}) \\\nonumber
&=&    \begin{cases}
        B' - c_t(A_{cs.tlm}) &\text{$A_{uw}=0$;} \\
        - c_t(A_{cs.tlm}) + d &\text{$A_{uw}>0$ and detected;}\\
        - Penalty - c_t(A_{cs.tlm}) &\text{$A_{uw}>0$ and not detected.}
    \end{cases}
\end{eqnarray}
If there is no faked computation (i.e., $A_{uw}=0$),
the CS.TLM's utility is $B' - c_t(A_{cs.tlm})$
where $B'$ is the benefit from having the UW in federated learning
and $C_t(x)$ is the cost for detecting $x$ randomly-selected computations.
If there is faked computation and it is detected,
the CS.TLM takes the UW's deposit and thus its utility is 
$- c_t(A_{cs.tlm}) + d$. 
If none of the $A_{uw}$ faked computation is detected,
the CS.TLM is penalized by $Penalty$ for the failure in detection and thus 
its utility is $- Penalty - c_t(A_{cs.tlm})$.

In the game, the goal of the CS and TLM coalition is to 
enforce an economically-greedy UW to execute all $n$ computations honestly. 
The following theorem states the conditions for the goal to be attained.

\begin{theorem}
For an economically-greedy untrusted local worker (UW) who aims to maximize its utility,
if the CS and TLM coalition's testing probability $\frac{A_{cs.tlm}}{n}>\frac{1}{n}$ (i.e., $A_{cs.tlm}>1$) 
and the UW's deposit $d\geq\frac{c}{1-e^{-(A_{cs.tlm}-1)}}$,
where $c$ is the cost for executing all the $n$ computations,
the UW should honestly execute all the $n$ computations.
\end{theorem}

\begin{proof}
We let $p=A_{cs.tlm}$ and $m=A_{uw}$ for convenience.
In the proof, 
we consider two separate cases: 
$1\leq m\leq (1-\frac{1}{p})n$ and $(1-\frac{1}{p})n<m\leq n$.\\

{\noindent \underline{Case I: $1\leq m\leq (1-\frac{1}{p})n$}}.
For this case, we prove by induction that probability for successful detection, 
i.e., $1-(1-\frac{p}{n})^m$, is at least $\frac{m}{n}$. 
That is:
\begin{equation}
    1-(1-\frac{p}{n})^m > \frac{m}{n}.
\end{equation}

{\em Base Case}. When $m=1$, 
\begin{equation}
    1-(1-\frac{p}{n})^m = \frac{p}{n} > \frac{1}{n} = \frac{m}{n}.
\end{equation}

{\em Inductive Step}. 
For any $1\geq m_0 \leq (1-\frac{1}{p})n-1$, 
we prove in the following that 
$1-(1-\frac{p}{n})^{m_0+1} > \frac{m_0+1}{n}$ 
as long as 
$1-(1-\frac{p}{n})^{m_0} > \frac{m_0}{n}$:

\begin{eqnarray}\small
&& (1-\frac{p}{n})^{m_0+1} %
= (1-\frac{p}{n})^{m_0}\cdot (1-\frac{p}{n}) \nonumber\\
&<& (1-\frac{m_0}{n})\cdot (1-\frac{p}{n}),~\mbox{for inductive assumption} \nonumber\\
&<& 1-\frac{m_0}{n} - \frac{1}{p}\cdot\frac{p}{n},~\mbox{because}~m_0\leq n(1-\frac{1}{p}) \nonumber\\
&=& 1-\frac{m_0+1}{n}.
\end{eqnarray}
Hence,
$1-(1-\frac{p}{n})^{m_0+1} > \frac{m_0+1}{n}$.

Given the above probability of successful detection,
the expected loss of deposit due to $m$ dishonest computations is at least
\begin{equation}
    \frac{m}{n}\cdot d \geq \frac{m}{n}\cdot \frac{c}{1-e^{-(p-1)}} > \frac{m\cdot c}{n}.
\end{equation}
That is, it is greater than the cost 
that can be saved by the UW who conducts $m\leq (1-\frac{1}{p})n$ dishonest computations. \\

{\noindent \underline{Case II: $(1-\frac{1}{p})n<m\leq n$}}. 
The probability can be derived as follows.

\begin{eqnarray}\small
&& 1-(1-\frac{p}{n})^m %
= 1-(1-\frac{p}{n})^{\frac{n}{p}\cdot\frac{m\cdot p}{n}} \nonumber\\
&>& 1- e^{-m\cdot\frac{p}{n}},~\mbox{because}~\forall x>0,~(1-\frac{1}{x})^x<e^{-1} \nonumber\\
&>& 1- e^{-n(1-\frac{1}{p})\frac{p}{n}},~\mbox{because}~m> (1-\frac{1}{p})n \nonumber\\
&=& 1- e^{-(p-1)}.
\end{eqnarray}

Given the above probability of successful detection,
the expected loss of deposit due to $m$ dishonest computations is at least
\begin{equation}
    (1- e^{-(p-1)})\cdot d \geq c.
\end{equation}
That is, it is greater than the cost 
that can be saved by the UW who fakes $m> (1-\frac{1}{p})n$ computations.

\end{proof}

{\noindent\bf{Remarks:}}~
Based on the above theorem, letting $p=2$, 
the UW is only required to make a deposit of $\frac{c}{1-e^{-1}}<2c$
and the TLM only needs to test 2 of the $n$ operations. 
When applying our proposed scheme,
$c$ is the maximal cost for executing any stage of the procedure
of training a neural network model,
which is small in practice.
Hence, our proposed scheme is practical.

\section{Enhancements for Higher Efficiency}

The general model presented in Section~\ref{sec:select-test-general-model}
could be directly applied for testing the procedures of 
forward or backward propagation through a convolutional or fully-connected layer. 
However, when the sizes of inputs and/or outputs are large, 
the costs for the UW to compute the Merkle trees and 
for the TLM to conduct selective testing could be very high. 
To address this problem,
we propose enhanced methods to attain higher efficiency.
Due to space limit, in the following
we present only 
backward propagation through a convolutional layer.
The techniques used here have also be applied to enhance the other procedures.

During the backward propagation through a convolutional layer,
the UW could compute the gradients for $\vec{X}$ 
according to Equation~(\ref{eq:conv-backward})
based on $\overrightarrow{\nabla Y}$ and the filters;
compute the updates to every filter 
according to Equation~(\ref{eq:conv-update})
based on $\overrightarrow{\nabla Y}$ and $\vec{X}$.
To facilitate selective testing efficiently,
however, we propose to make changes to the above procedure. 
We elaborate the new algorithms for computation and 
selective testing in the following.

\subsection{Computation by the UW}

To facilitate selective test,
the TLM should record some intermediate results of the computations 
and construct several Merkel hash trees.  
After the trees have been constructed,
the UW send their commitments to the TLM before 
the TLM conducts selective testing.

\subsubsection{Computing $\protect\overrightarrow{\nabla X}$ and $Tree(\protect\overrightarrow{\nabla X})$}

For each filter $\overrightarrow{F^{(t)}}$ with $t\in[n_F]$, 
the TLM computes matrix $\overrightarrow{\nabla X^{(t)}}$
based on Equation~(\ref{eq:conv-backward-1}) and then
computes matrix $\overrightarrow{\nabla X}$ based on Equation~(\ref{eq:conv-backward}).
The TLM records 
The resulting $\{\overrightarrow{\nabla X^{(t)}}|t\in[n_R]\}$ and 
$\overrightarrow{\nabla X}$, and further use them 
to construct a Merkle hash tree $Tree(\overrightarrow{\nabla X})$ as follows.
For each $i,j\in[\alpha_X]$,
the TLM computes 
$h(\overrightarrow{\nabla X},i,j)=hash(\overrightarrow{\nabla X^{(1)}}[i,j],\cdots,\overrightarrow{\nabla X^{(n_F)}}[i,j])$,
uses all of these hash values as leaf nodes to construct $Tree(\overrightarrow{\nabla X})$,
and uses the root of the tree as commitment $comm(\overrightarrow{\nabla X})$.

\subsubsection{Computing 
$\protect\overrightarrow{\nabla F^{(t)}}$ and $Tree(\protect\overrightarrow{\nabla F})$}

To record some intermediate results of computing $\overrightarrow{\nabla F^{(t)}}$, 
Each element of it, i.e., 
$\overrightarrow{\nabla F^{(t)}}[i,j]$ for every $i,j\in[\alpha_F]$, 
is expanded to a $\alpha_Y$-element vector
denoted as $\overrightarrow{\nabla F^{(t)}_{i,j}}$ and
each element of the vector, %
which is denoted as $\overrightarrow{\nabla F^{(t)}_{i,j}}[u]$ for $u\in[\alpha_Y]$,
is computed as
\begin{equation}\label{eq:nabla_Ft_iju}
    \overrightarrow{\nabla F^{(t)}_{i,j}}[u]=
        \sum_{v\in[\alpha_Y]} \overrightarrow{\nabla Y^{(t)}}[u,v]\cdot \vec{X}[(u-1)\alpha_F+i,(v-1)\alpha_F+j].
\end{equation}
Then, $\overrightarrow{\nabla F^{(t)}}[i,j]$ is computed as
\begin{equation}
    \overrightarrow{\nabla F^{(t)}}[i,j] = 
    \sum_{u\in [\alpha_Y]}\overrightarrow{\nabla F^{(t)}_{i,j}}[u].
\end{equation}
For each vector $\overrightarrow{\nabla F^{(t)}_{i,j}}$ with $t\in [n_F]$ and $i,j\in[\alpha_Y]$,
the hash of all its elements is computed.
Then, such hash values for every $t$, $i$ and $j$
are used as leaf node to construct Merkel hash tree
$Tree(\overrightarrow{\nabla F})$.

\subsubsection{Constructing 
$Tree(\protect\overrightarrow{\nabla Y})$ and $Tree'(\protect\overrightarrow{X})$}

Based on how $\overrightarrow{\nabla Y}$ and $\overrightarrow{X}$
are used in Equation~(\ref{eq:nabla_Ft_iju}),
the TLM further constructs the following two Merkel hash trees.

For each $t\in[n_F]$, a hash value is computed for each of 
the $\alpha_Y$ rows of $\overrightarrow{\nabla Y^{(t)}}$.
Then, all these $\alpha_Y\cdot n_F$ hash values are used as leaf nodes
to construct Merkel hash tree $Tree(\overrightarrow{\nabla Y})$.

For each $i,j\in[\alpha_F]$ and each $u\in[\alpha_Y]$,
we define a vector denoted as $\vec{X}_{i,j,u}$ that includes 
the following elements of $\vec{X}$: 
$\vec{X}[(u-1)\alpha_F+i,(v-1)\alpha_F+j]$ for every $v\in[\alpha_Y]$.
The hash for all the elements in $\vec{X}_{i,j,u}$, 
denoted as $h(\vec{X}_{i,j,u})$, is computed.
Then, all of the above hash values  
$h(\vec{X}_{i,j,u})$ for every $u\in[n_F]$
are used as leaf nodes to construct Merkel hash tree $Tree'(\overrightarrow{X})$.

\subsection{Selective Testing by the TLM}

The TLM issues a request to the UW for pointers to the memory where 
$\vec{X}$, 
$\overrightarrow{\nabla Y}$, 
$\overrightarrow{\nabla F}$, 
$\overrightarrow{\nabla X}$,
$\overrightarrow{\nabla X^{(t)}}$ for every $t\in[n_F]$,
$Tree'(\vec{X})$, 
$Tree(\overrightarrow{\nabla Y})$,
$Tree(\overrightarrow{\nabla F})$, and 
$Tree(\overrightarrow{\nabla X})$ 
are stored.
After receiving the information,
the TLM selectively tests the computations of
$\overrightarrow{\nabla X}$ and 
$\overrightarrow{\nabla F}$
as follows.
Note that, the testings use the filters and 
we assume the filters, 
due to their small size,
are kept in the trusted memory space of the TLM.

\subsubsection{Selectively Testing $\protect\overrightarrow{\nabla X}$}

The TLM randomly selects $p$, 
which should be greater than $1$ according to Theorem 1, 
elements of $\overrightarrow{\nabla X}$ to test.
For each selected element,
denoted as $\overrightarrow{\nabla X^{(t)}}[i,j]$ for 
certain $t\in[n_R]$ and $i,j\in[\alpha_X]$, 
the testing is as follows.

{\noindent\bf{Test~1}}:~
$\overrightarrow{\nabla X^{(t)}}[i,j]$ 
is validated based on $Tree(\overrightarrow{\nabla X})$ and 
the commitment $comm(Tree(\overrightarrow{\nabla X}))$ 
that the TLM has received earlier from the UW.
Specifically, the TLM computes  
$h=hash(\overrightarrow{\nabla X^{(1)}}[i,j],\cdots,\overrightarrow{\nabla X^{(n_F)}}[i,j])$;
checks whether $h$ is equal to the leaf node of $Tree(\vec{X})$
at the position (denoted as $index$) corresponding to element $\overrightarrow{\nabla X^{(t)}}[i,j]$;
retrieves the co-path values of the leaf node from $Tree(\vec{X})$ to form evidence $evid$;
and finally calls $Verify\_Commit(h,index,comm(Tree(\overrightarrow{\nabla X})),evid)$ 
to verify the validity.

{\noindent\bf{Test~2}}:~
The TLM identifies the elements of $\overrightarrow{\nabla Y^{(t)}}$ 
that are used in computing $\overrightarrow{\nabla X^{(t)}}[i,j]$ 
and validates these elements based on 
$Tree(\overrightarrow{\nabla Y})$ and 
the commitment $comm(Tree(\overrightarrow{\nabla Y}))$ received earlier.
According to Equation~(\ref{eq:conv-backward-1}), 
these elements include every $\overrightarrow{Y^{(t)}}[u,v]$ such that 
$u,v\in[\alpha_Y]$ and $i-u,j-v\in[n_F]\}$.
Also, these elements belong to every row $r$ of matrix $\overrightarrow{Y^{(t)}}$
such that $u\in[\alpha_Y]$ and $i-u\in[\alpha_F]$;
note that, the number of such rows is at most $\alpha_F$.
Thus, the TLM should retrieve all the elements in these rows. 
For each row $u$, it computes the hash value $h$ of all the elements in the row,
and checks if $h$ equals to the leaf node of $Tree(\overrightarrow{\nabla Y})$
with index $index$ that corresponds to the row. 
If so, the co-path hash values of the leaf are retrieved and recorded as evidence $evid$, and
$Verify\_Commit(h,index,comm(Tree(\overrightarrow{\nabla Y})),evid)$ is called to 
verify the validity of the row.

{\noindent\bf{Test~3}}:~
Lastly, the TLM re-computes $\overrightarrow{\nabla X^{(t)}}[i,j]$ 
according to Equation~(\ref{eq:conv-backward-1})
and checks if the re-computed result is equal to 
the $\overrightarrow{\nabla X^{(t)}}[i,j]$ that was already verified in Test~1.

\subsubsection{Selectively Testing $\protect\overrightarrow{\nabla F}$}

The TLM randomly selects $p$ elements of $\overrightarrow{\nabla F}$ to test.
For each selected element,
denoted by $\overrightarrow{\nabla F^{(t)}_{i,j}}[u]$ for 
some $t\in[n_R]$ and $i,j,u\in[\alpha_F]$, 
the testing is as follows.

{\noindent\bf{Test~1}}:~
$\overrightarrow{\nabla F^{(t)}_{i,j}}[u]$ is validated 
based on $Tree(\overrightarrow{\nabla F})$ and commitment $commit(Tree(\overrightarrow{\nabla F}))$.
Specifically, 
the hash $h$ of all elements in vector $\overrightarrow{\nabla F^{(t)}_{i,j}}$ is computed;
the index $index$ of the leaf node corresponding to $h$ in $Tree(\overrightarrow{\nabla F})$ is identified;
the co-path values for the leaf node are identified in $Tree(\overrightarrow{\nabla F})$ to form evidence $evid$;
$Verify\_Commit(h,ind,comm(\overrightarrow{\nabla F},evid)$ is called to 
verify the validity of $\overrightarrow{\nabla F^{(t)}_{i,j}}$ and thus 
the validity of its element $\overrightarrow{\nabla F^{(t)}_{i,j}}[u]$.

{\noindent\bf{Test~2}}:~
$\overrightarrow{\nabla Y^{(t)}}[u]$, 
which is the row of matrix $\overrightarrow{\nabla Y^{(t)}}$ 
used in computing $\overrightarrow{\nabla F^{(t)}_{i,j}}[u]$ according to Equation~(\ref{eq:nabla_Ft_iju}),
is validated based on $Tree(\overrightarrow{\nabla F})$ and commitment $comm(Tree(\overrightarrow{\nabla F}))$. 
Specifically,
the hash $h$ of all elements in $\overrightarrow{\nabla Y^{(t)}}[u]$ is computed;
the index $index$ of the leaf node corresponding to $h$ in $Tree(\overrightarrow{\nabla Y})$ is identified;
the co-path values for the leaf node are identified in $Tree(\overrightarrow{\nabla Y})$ to form evidence $evid$;
$Verify\_Commit(h,ind,comm(Tree(\overrightarrow{\nabla F})),evid)$ is called to 
verify the validity of $\overrightarrow{\nabla Y^{(t)}}[u]$.

{\noindent\bf{Test~3}}:~
$\vec{X}_{i,j,u}$,
which is the group of elements in $\vec{X}$ that are used in computing $\overrightarrow{\nabla F^{(t)}_{i,j}}[u]$,
is validated based on $Tree'(\vec{X})$ and its commitment $comm(Tree'(\vec{X}))$. 
Specifically,
the hash $h$ of all elements in group $\vec{X}_{i,j,u}$ is computed;
the index $index$ of the leaf node corresponding to $h$ in $Tree'(\vec{X})$ is identified;
the co-path values for the leaf node are identified in $Tree'(\vec{X})$ to form evidence $evid$;
$Verify\_Commit(h,index,comm(Tree'(\vec{X})),evid)$ is called to 
verify the validity of group $\vec{X}_{i,j,u}$.

{\noindent\bf{Test~4}}:~
Lastly, the TLM re-computes $\overrightarrow{\nabla F^{(t)}_{i,j}}[u]$ according to Equation~(\ref{eq:nabla_Ft_iju})
and checks if the re-computed result equals to $\overrightarrow{\nabla F^{(t)}_{i,j}}[u]$
which is already verified in Test~1.

\section{Performance Evaluation}

For performance evaluation,
we implement our proposed new scheme on a computer with Intel SGX.
We also implement the following schemes for comparison:
    {\em Original (No-SGX) Scheme} - the untrusted server implements 
    the convoluntional and fully-connected layer functions without any security consideration.
    {\em Full-SGX Scheme} - the SGX enclave implements the convolutional and fully-connected
    layer functions. 
Note that, 
for the full-SGX scheme, 
due to limited trusted memory space,
data should be loaded from the regular memory to the enclave before being processed 
and the processing results should be stored back to the regular memory.
To ensure the integrity of the data,
a hash value of the data is computed and stored securely in enclave 
before the data is stored to the regular memory; the hash is recomputed
and compared to the stored hash when the data is re-loaded to the enclave.
    
The above three schemes are evaluated on a computer 
with Intel Core i5-8400 CPU (2.80GHz) of six cores and a RAM of 8.00GB. 
The evaluation results are presented and discussed in the following. 

\subsubsection*{Convolutional Layer: Forward Propagation}

Table~\ref{tab:cc-fwd-xsize} shows the costs of the schemes
for the forward propagation through a convolutional layer,
as the input size varies.
The original scheme's cost is denoted as {\em original fwd} and
the full-SGX scheme's cost is denoted as {\em SGX fwd}.
For our proposed scheme, 
the cost incurred at the untrusted worker is dentoed as {\em new fwd}
and the cost for selective test incurred at the SGX enclave is denoted as {\em selective test}.
All the costs are measured as the computation latency in the unit of micro-second. 
Here, %
16 filters each of size 8$\times$8 are used and the stride is set to 2.

\begin{table}[htb]\small
    \centering
    \begin{tabular}{c|c|c|c|c}
    \hline
    input size & original fwd & SGX fwd & new fwd & selective test \\\hline
    16$\times$16   & 124    & 145   & 303    & 35  \\
    32$\times$32   & 818    & 865   & 1265   & 41  \\
    64$\times$64   & 3990   & 4161  & 5242   & 59  \\
    128$\times$128 & 17705  & 18382 & 21065  & 118 \\
    256$\times$256 & 74196  & 76868 & 84815  & 360 \\\hline
    \end{tabular}
    \caption{Forward Propagation Costs for Convolutional Layer 
    (unit: micro-second): Impact of Input Size. 
    Settings: $stride=2$, $filter\_size=8\times 8$, and $filter\_number=16$.}
    \label{tab:cc-fwd-xsize}
\end{table}

As we can see from Table~\ref{tab:cc-fwd-xsize}, 
the cost of the original scheme is slightly lower than the full-SGX scheme
due to the extra overhead for ensuring data integrity.
Our new scheme introduces higher cost at the untrusted worker,
at the price of significantly reducing the cost at the SGX enclave.
The results also demonstrate that, when the input size is not small (i.e., 
greater than 32$\times$32), 
the new scheme does not increase the cost of the untrusted worker significantly 
(i.e., 1.14-1.55 times of the original scheme)
while incurring significantly lower cost at the SGX enclave 
(i.e., 0.5\%-4.7\% of the full-SGX scheme).

\begin{table}[htb]\small
    \centering
    \begin{tabular}{c|c|c|c|c}
    \hline
    filter number & original fwd & fwd by SGX & new fwd & selective test \\\hline
    4  & 4411  & 4666  & 5275  & 111 \\
    8  & 8866  & 9201  & 10577 & 117 \\
    16 & 17626 & 18587 & 20939 & 111 \\
    32 & 35367 & 36730 & 41942 & 115 \\\hline
    \end{tabular}
    \caption{Forward Propagation Costs for Convolutional Layer 
    (unit: micro-second): Impact of Output Size. %
    Settings: $stride=2$, $filter\_size=8\times8$, and $input\_size=128\times128$.}
    \label{tab:cc-fwd-fnum}
\end{table}

\begin{table}[htb]
    \centering\small
    \begin{tabular}{c|c|c|c|c}
    \hline
    stride & original fwd & fwd by SGX & new fwd & selective test \\\hline
    1  & 68941 & 71575  & 78447  & 123 \\
    2  & 17626 & 18587  & 20939  & 111 \\
    4  & 4583  & 4789   & 6009   & 110 \\
    8  & 1215  & 1297   & 1857   & 107 \\\hline
    \end{tabular}
    \caption{Forward Propagation Costs for Convolutional Layer 
    (unit: micro-second): Impact of stride. 
    Settings: $filter\_size=8\times8$, $filter\_number=16$ and $input\_size=128\times128$.}
    \label{tab:cc-fwd-stride}
\end{table}

\begin{table}[htb]
    \centering
    \begin{tabular}{c|c|c|c|c}
    \hline
    filter size & original fwd & fwd by SGX & new fwd & selective test \\\hline
    8$\times$8   & 17708   & 18449   & 21117   & 113  \\
    16$\times$16 & 61104   & 53713   & 64606   & 132  \\
    32$\times$32 & 179681  & 149620  & 182581  & 149  \\
    64$\times$64 & 319404  & 268115  & 320935  & 152  \\\hline
    \end{tabular}
    \caption{Forward Propagation Costs for Convolutional Layer 
    (unit: micro-second): Impact of filter size. 
    Settings: $stride=2$, $filter\_number=16$ and $input\_size=128\times128$.}
    \label{tab:cc-fwd-fsize}
\end{table}

Similar trends have been demonstrated in Tables~\ref{tab:cc-fwd-fnum},
\ref{tab:cc-fwd-stride}, and \ref{tab:cc-fwd-fsize},
where the costs incurred by the three schemes are presented
as the number/size of the filters or the stride changes.
Specifically, 
the costs of the original and the full-SGX schemes are similar,
the new scheme introduces a slightly higher cost at the untrusted worker
(i.e., 1.01-1.53 times of the original scheme)
and incurs much lower cost at the SGX enclave
(i.e., 0.1\%-8.2\% of the full-SGX scheme).

\subsubsection*{Convolutional Layer: Backward Propagation}

Table~\ref{tab:cc-bwd-xsize} shows the costs of the three schemes
for backward propagation through a convolutional layer,
as the size of the input varies from $16\times16$ to $256\times256$.
According to the table, 
the full-SGX scheme's cost (denoted as {\em SGX bwd}) is higher than (i.e., about twice of) 
the original scheme's cost (denoted as {\em original bwd}),
because the full-SGX scheme needs to load and check the integrity of 
the inputs and outputs of the layer.
The new scheme's cost at the untrusted worker (denoted as {\em new bwd}) is also high
because the worker needs to construct large Merkle hash trees 
to facilitate selective testing. Specifically, 
when the input size is not large (i.e., $32\times32$ or smaller),
the cost at the worker is as high as 5-24 times of the original scheme's cost.
However, when the input size becomes larger than $64\times 64$, 
the worker's cost becomes only 1.5-2.3 times of the original scheme's cost.
Particularly, the worker's cost is even smaller than the full-SGX scheme's cost
when the input size is $128\times 128$ or larger. 
The new scheme's cost at the SGX enclave (denoted as {\em selective test}) remains the smallest;
it is 6-28\% of the full-SGX scheme's cost when the input size is no greater than $32\times 32$
and only 0.3-1.6\% of the full-SGX scheme's cost when the input size is $64\times64$ or larger. 

\begin{table}[htb]
    \centering
    \begin{tabular}{c|c|c|c|c}
    \hline
    input size & original bwd & SGX bwd & new bwd & selective test \\\hline
    16$\times$16   & 188    & 398    & 4611   & 112  \\
    32$\times$32   & 1249   & 2516   & 6254   & 139  \\
    64$\times$64   & 6058   & 12049  & 13768  & 191  \\
    128$\times$128 & 26771  & 53107  & 43840  & 305 \\
    256$\times$256 & 112449 & 223482 & 173460 & 658 \\\hline
    \end{tabular}
    \caption{Backward Propagation Costs for Convolutional Layer 
    (unit: micro-second): Impact of Input Size. 
    Settings: $stride=2$, $filter\_size=8X8$, and $filter\_number=16$.}
    \label{tab:cc-bwd-xsize}
\end{table}

\begin{table}[htb]
    \centering
    \begin{tabular}{c|c|c|c|c}
    \hline
    filter number & original bwd & SGX bwd & new bwd & selective test \\\hline
    4  & 6702       &    13724&          11560&          228 \\
    8  & 13426      &    26874&          22327&          253 \\
    16 & 26765      &    52957&          43884&          292 \\
    32 & 53520      &    105403&         93123&          389 \\\hline
    \end{tabular}
    \caption{Backward Propagation Costs for Convolutional Layer 
    (unit: micro-second): Impact of $filter\_number=16$. 
    Settings: $stride=2$, $filter\_size=8X8$, and $input\_size=128\times128$.}
    \label{tab:cc-bwd-fnum}
\end{table}

\begin{table}[htb]
    \centering
    \begin{tabular}{c|c|c|c|c}
    \hline
    stride & original bwd & SGX bwd & new bwd & selective test \\\hline
    1  & 104705&         206806&         130990&         367 \\
    2  & 26765&          52957&          43884&          292 \\
    4  & 6994&           14308&          21471&          272 \\
    8  & 1868&           4270&           15292&          259 \\\hline
    \end{tabular}
    \caption{Backward Propagation Costs for Convolutional Layer 
    (unit: micro-second): Impact of stride. 
    Settings: $filter\_size=8X8$, $filter\_number=16$ and $input\_size=128\times128$.}
    \label{tab:cc-bwd-stride}
\end{table}

\begin{table}[htb]
    \centering
    \begin{tabular}{c|c|c|c|c}
    \hline
    filter size & original bwd & SGX bwd & new bwd & selective test \\\hline
    8$\times$8   & 26777&          53085&          43984&          291  \\
    16$\times$16 & 92652&          184529&         143743&         342  \\
    32$\times$32 & 269609&         541327&         447952&         420  \\
    64$\times$64 & 488914&         981685&         1037886&        446  \\\hline
    \end{tabular}
    \caption{Backward Propagation Costs for Convolutional Layer 
    (unit: micro-second): Impact of filter size. 
    Settings: $stride=2$, $filter\_number=16$ and $input\_size=128\times128$.}
    \label{tab:cc-bwd-fsize}
\end{table}

Similar trends can be observed in Tables~\ref{tab:cc-bwd-fnum}, 
\ref{tab:cc-bwd-stride} and \ref{tab:cc-bwd-fsize},
where the schemes' costs are compared 
as the number/size of filters or the stride changes
but the input size is fixed at $128\times 128$. 
Specifically, the new scheme's cost at the untrusted worker ranges between 1.25-3.07 times of the original scheme's cost,
except that the cost is 8.19 times of the original scheme's cost when the filter size is $8\times 8$ and stride is $8$,
in which case the original scheme's workload is small because the stride is large relative to the filter size.
The new scheme's cost at the SGX enclave remains low; specifically,
it ranges between 0.05-6.1\% of the full-SGX scheme's cost.

\subsubsection*{Fully-connected Layer: Forward Propagation}

\begin{table}[htb]
    \centering
    \begin{tabular}{c|c|c|c|c}
    \hline
    input size  & original fwd & SGX fwd & new fwd & selective test \\\hline
    32&     6&              114&            134&            33  \\
    64&     11&             217&            148&            37  \\
    128&    21&             410&            144&            36  \\
    256&    43&             781&            201&            36  \\
    512&    92&             1563&           270&            38  \\
    1024&   205&            3030&           413&            37  \\
    2048&   481&            6050&           715&            39  \\
    4096&   788&            12128&          1196&           41  \\\hline
    \end{tabular}
    \caption{Forward Propagation Costs for Fully-connected Layer 
    (unit: micro-second): Impact of Input Size.
    Here the output size is $64$.}
    \label{tab:cc-fwd-input}
\end{table}

\begin{table}[htb]
    \centering
    \begin{tabular}{c|c|c|c|c}
    \hline
    output size & original fwd & SGX fwd & new fwd & selective test \\\hline
    4096&   191713&         1600367&        249100&         56  \\
    2048&   102136&         786005&         128956&         56  \\
    1024&   51145&          383069&         65157&          53  \\
    512&    25597&          183947&         33951&          60  \\
    256&    12820&          96058&          16649&          55  \\
    128&    4529&           32506&          5793&           56  \\
    64&     1245&           12599&          1443&           48  \\
    32&     479&            6251&           636&            43  \\
    16&     208&            3092&           284&            43  \\\hline
    \end{tabular}
    \caption{Forward Propagation Costs for Fully-connected Layer 
    (unit: micro-second): Impact of Output Size.
    Here the input size is $4096$.}
    \label{tab:cc-fwd-output}
\end{table}

Tables~\ref{tab:cc-fwd-input} and \ref{tab:cc-fwd-output} show 
the costs of the three schemes for
the forward propagation through a fully-connected layer,
as the input and output sizes vary. 
As we can see, the full-SGX scheme has higher cost than 
the original cost due to the extra overheads 
for loading and verifying the integrity of the weight matrix, 
the size of which increases along with the input or output size,
and for computing the hash of the outputs. 

Except for the cases when the input and output sizes are small (e.g., 
input size is no greater than 32 and the output size is no greater than 64), 
the new scheme has lower cost at the untrusted worker than the full-SGX scheme. 
Specifically, the untrusted worker's cost ranges between 9-69\% of the full-SGX scheme's cost. 

The new scheme's cost at the SGX enclave (i.e., selective test) remains the smallest.
Table~\ref{tab:cc-fwd-input} shows that,
as the output size increases from $32$ to $4096$,
the cost for selective test increases only slightly from $33$ to $41$ micro-seconds
while the full-SGX scheme's cost increases by $106$ times;
therefore, the cost for selective test changes from 29\% to 0.3\% of 
the full-SGX scheme's cost. 
Similarly, Table~\ref{tab:cc-fwd-output} shows that,
as the output size increases from $16$ to $4096$,
the cost for selective test increases only slightly from $43$ to $56$ micro-seconds
while the full-SGX scheme's cost increases by $517$ times;
therefore, the cost for selective test changes from 1.3\% to 0.003\% of
the full-SGX scheme's cost.

\subsubsection*{Fully-connected Layer: Backward Propagation}

\begin{table}[htb]
    \centering
    \begin{tabular}{c|c|c|c|c}
    \hline
    input size  & original bwd & SGX bwd & new bwd & selective test \\\hline
    32&     10&             198&            202&            33  \\
    64&     20&             388&            393&            37  \\
    128&    39&             752&            659&            36  \\
    256&    74&             1484&           1141&           41  \\
    512&    151&            2955&           2468&           45  \\
    1024&   293&            5824&           4132&           50  \\
    2048&   585&            11609&          8305&           70  \\
    4096&   1178&           23215&          16437&          102  \\\hline
    \end{tabular}
    \caption{Backward Propagation Costs for Fully-connected Layer 
    (unit: micro-second): Impact of Input Size.
    Here the output size is $64$.}
    \label{tab:cc-bwd-input}
\end{table}

\begin{table}[htb]
    \centering
    \begin{tabular}{c|c|c|c|c}
    \hline
    output size & original bwd & SGX bwd & new bwd & selective test \\\hline
    4096&   80876&          1461079&        503361&         152  \\
    2048&   40225&          730257&         257843&         154  \\
    1024&   20426&          364780&         134487&         150  \\
    512&    10620&          186273&         73545&          113  \\
    256&    5234&           93183&          41460&          110  \\
    128&    2706&           46746&          24964&          108  \\
    64&     1444&           23394&          17034&          108  \\
    32&     701&            11780&          12222&          106  \\
    16&     298&            6028&           10163&          103  \\\hline
    \end{tabular}
    \caption{Backward Propagation Costs for Fully-connected Layer 
    (unit: micro-second): Impact of Output Size.
    Here the input size is $4096$.}
    \label{tab:cc-bwd-output}
\end{table}

Tables~\ref{tab:cc-bwd-input} and \ref{tab:cc-bwd-output} show 
the costs of the three schemes for
the backward propagation through a fully-connected layer,
as the input and output sizes vary. 
The trends are similar to those shown in 
Tables~\ref{tab:cc-fwd-input} and \ref{tab:cc-fwd-output}.

\section{Related Works}

There have been many schemes devised in order to provide for private deep learning \cite{Mirshghallah2020PrivacyID,Chabanne2017PrivacyPreservingCO,Tramr2019SlalomFV,Bu2020DeepLW,Reagen2021CheetahOA}. The research commonly uses statistical, cryptographic, and hardware techniques in order to achieve this.
Differential privacy is a statistical technique that has been used in the data aggregation, training phase, and inference phases\cite{Mirshghallah2020PrivacyID}. Amongst the challenges presented by using this technique is maximizing privacy while minimizing loss of accuracy \cite{Bu2020DeepLW}.
One cryptographic approach for providing privacy during the inference \cite{Dowlin2016CryptoNetsAN} \cite{Boddeti2018SecureFM} and training phases is homomorphic encryption. Some research \cite{Reagen2021CheetahOA}, shows methods for using homomorphic encryption to protect the model, while others for protecting the data. In both cases, maintaining high performance, or throughput, is a persistent challenge. In order to apply activations such a ReLU to encrypted data, techniques such as using polynomial approximations with batch normalization have been developed \cite{Reagen2021CheetahOA}.
The hardware approach often involves using multiparty computation or trusted execution environments. Tramèr and Boneh \cite{Tramr2019SlalomFV} make use of TEEs to allow inference that protects the privacy of input data. Furthermore, their scheme provides integrity, and still allows for outsourcing linear operations to an untrusted external GPU. Their framework also takes advantage of the fact that matrix multiplication can be verified asymptotically more efficiently than it can be computed \cite{Freivalds1977ProbabilisticMC}.

The distributed nature of federated learning introduces new security concerns. 
Particularly, it may be possible for a curious server to infer information about the data used by clients 
during the training process. 
Secure aggregation \cite{Bonawitz2017PracticalSA,Fereidooni2021SAFELearnSA, Phong2018PrivacyPreservingDL} is an attempt to prevent this by ensuring that no party reveals its individual updates in the clear. 
For instance, VerifyNet \cite{Xu2020VerifyNetSA} builds upon the secure aggregation of 
PPML \cite{Bonawitz2017PracticalSA} while also providing the ability for participating clients to verify that the server performed the aggregation correctly.

Another potential threat in the federated setting comes from data poisoning. Clients could attempt to poison the global model by injecting maliciously labeled data before the learning starts. One approach to combat this is using more sophisticated aggregation rules \cite{Blanchard2017MachineLW} \cite{Yin2018ByzantineRobustDL}. Malicious clients may be able to circumvent the protections of Byzantine-robust aggregation rules by maliciously labeling data during the training phase, causing the model to have a large error rate once trained \cite{Fang2020LocalMP}. It is also possible for participants to engage in targeted attacks, which seek to impact classification for only specific classes, with other classes remaining largely unaffected.  \cite{Tolpegin2020DataPA} proposes a method of identifying these malicious participants, having the aggregating server perform PCA on the parameter updates received from participating clients.

Finally, clients may wish to receive credit for participating in the training without actually doing to the training that is expected of them, which little attention has been paid to defending against. This paper aims to fill this gap by
proposing a scheme to ensure local workers' honest execution of local learning based on 
the TEE technology, game theory and applied cryptography.

\section{Conclusion and Future Work}

In this paper, 
we proposed a game-theoretic and TEE-based scheme 
to ensure the correctness of computations performed by an untrusted worker in a federated learning system. 
Through smart contract and selectively choosing which untrusted computations to test, 
computational overhead performed by the TEE is minimal, 
drastically reduced when compared to the baseline schemes. 
In the future, it may be possible to expand the scheme to more kinds of neural networks. 
It may also be possible to improve the performance of the commitment process 
by using alternative cryptographic constructions.

\bibliographystyle{IEEEtran}
\bibliography{reference2}

\begin{thebibliography}{10}
\providecommand{\url}[1]{#1}
\csname url@samestyle\endcsname
\providecommand{\newblock}{\relax}
\providecommand{\bibinfo}[2]{#2}
\providecommand{\BIBentrySTDinterwordspacing}{\spaceskip=0pt\relax}
\providecommand{\BIBentryALTinterwordstretchfactor}{4}
\providecommand{\BIBentryALTinterwordspacing}{\spaceskip=\fontdimen2\font plus
\BIBentryALTinterwordstretchfactor\fontdimen3\font minus
  \fontdimen4\font\relax}
\providecommand{\BIBforeignlanguage}[2]{{%
\expandafter\ifx\csname l@#1\endcsname\relax
\typeout{** WARNING: IEEEtran.bst: No hyphenation pattern has been}%
\typeout{** loaded for the language `#1'. Using the pattern for}%
\typeout{** the default language instead.}%
\else
\language=\csname l@#1\endcsname
\fi
#2}}
\providecommand{\BIBdecl}{\relax}
\BIBdecl

\bibitem{Gubbi2013InternetOT}
J.~Gubbi, R.~Buyya, S.~Marusic, and M.~Palaniswami, ``Internet of things (iot):
  A vision, architectural elements, and future directions,'' \emph{ArXiv}, vol.
  abs/1207.0203, 2013.

\bibitem{Esposito2018BlockchainAP}
C.~Esposito, A.~D. Santis, G.~Tortora, H.~Chang, and K.-K.~R. Choo,
  ``Blockchain: A panacea for healthcare cloud-based data security and
  privacy?'' \emph{IEEE Cloud Computing}, vol.~5, pp. 31--37, 2018.

\bibitem{McMahan2017CommunicationEfficientLO}
H.~McMahan, E.~Moore, D.~Ramage, S.~Hampson, and B.~A. y~Arcas,
  ``Communication-efficient learning of deep networks from decentralized
  data,'' in \emph{AISTATS}, 2017.

\bibitem{Bonawitz2017PracticalSA}
K.~Bonawitz, V.~Ivanov, B.~Kreuter, A.~Marcedone, H.~B. McMahan, S.~Patel,
  D.~Ramage, A.~Segal, and K.~Seth, ``Practical secure aggregation for
  privacy-preserving machine learning,'' \emph{Proceedings of the 2017 ACM
  SIGSAC Conference on Computer and Communications Security}, 2017.

\bibitem{Fereidooni2021SAFELearnSA}
H.~Fereidooni, S.~Marchal, M.~Miettinen, A.~Mirhoseini, H.~M{\"o}llering,
  T.~Nguyen, P.~Rieger, A.~Sadeghi, T.~Schneider, H.~Yalame, and S.~Zeitouni,
  ``Safelearn: Secure aggregation for private federated learning,'' \emph{2021
  IEEE Security and Privacy Workshops (SPW)}, pp. 56--62, 2021.

\bibitem{Xu2020VerifyNetSA}
G.~Xu, H.~Li, S.~Liu, K.~Yang, and X.~Lin, ``Verifynet: Secure and verifiable
  federated learning,'' \emph{IEEE Transactions on Information Forensics and
  Security}, vol.~15, pp. 911--926, 2020.

\bibitem{Phong2018PrivacyPreservingDL}
L.~T. Phong, Y.~Aono, T.~Hayashi, L.~Wang, and S.~Moriai, ``Privacy-preserving
  deep learning via additively homomorphic encryption,'' \emph{IEEE
  Transactions on Information Forensics and Security}, vol.~13, pp. 1333--1345,
  2018.

\bibitem{IntelSGXexplained}
V.~Costan and S.~Devadas, ``Intelsgxexplained,'' \emph{IACR Cryptology
  ePrintArchive}, pp. 1--118, 2016.

\bibitem{trustzone}
``{Arm TrustZone Technology},''
  \url{https://developer.arm.com/ip-products/security-ip/trustzone}, [Online;
  accessed 1-August-2021].

\bibitem{Mirshghallah2020PrivacyID}
F.~Mirshghallah, M.~Taram, P.~Vepakomma, A.~Singh, R.~Raskar, and
  H.~Esmaeilzadeh, ``Privacy in deep learning: A survey,'' \emph{ArXiv}, vol.
  abs/2004.12254, 2020.

\bibitem{Chabanne2017PrivacyPreservingCO}
H.~Chabanne, A.~de~Wargny, J.~Milgram, C.~Morel, and E.~Prouff,
  ``Privacy-preserving classification on deep neural network,'' \emph{IACR
  Cryptol. ePrint Arch.}, vol. 2017, p.~35, 2017.

\bibitem{Tramr2019SlalomFV}
F.~Tram{\`e}r and D.~Boneh, ``Slalom: Fast, verifiable and private execution of
  neural networks in trusted hardware,'' \emph{ArXiv}, vol. abs/1806.03287,
  2019.

\bibitem{Bu2020DeepLW}
Z.~Bu, J.~Dong, Q.~Long, and W.~J. Su, ``Deep learning with gaussian
  differential privacy,'' \emph{Harvard data science review}, vol. 2020 23,
  2020.

\bibitem{Reagen2021CheetahOA}
B.~Reagen, W.~Choi, Y.~Ko, V.~T. Lee, H.-H.~S. Lee, G.-Y. Wei, and D.~Brooks,
  ``Cheetah: Optimizing and accelerating homomorphic encryption for private
  inference,'' \emph{2021 IEEE International Symposium on High-Performance
  Computer Architecture (HPCA)}, pp. 26--39, 2021.

\bibitem{Dowlin2016CryptoNetsAN}
N.~Dowlin, R.~Gilad-Bachrach, K.~Laine, K.~E. Lauter, M.~Naehrig, and
  J.~Wernsing, ``Cryptonets: applying neural networks to encrypted data with
  high throughput and accuracy,'' in \emph{ICML 2016}, 2016.

\bibitem{Boddeti2018SecureFM}
V.~N. Boddeti, ``Secure face matching using fully homomorphic encryption,'' in
  \emph{BTAS}, 2018.

\bibitem{Freivalds1977ProbabilisticMC}
R.~Freivalds, ``Probabilistic machines can use less running time,'' in
  \emph{IFIP Congress}, 1977.

\bibitem{Blanchard2017MachineLW}
P.~Blanchard, E.~M.~E. Mhamdi, R.~Guerraoui, and J.~Stainer, ``Machine learning
  with adversaries: Byzantine tolerant gradient descent,'' in \emph{NIPS},
  2017.

\bibitem{Yin2018ByzantineRobustDL}
D.~Yin, Y.~Chen, K.~Ramchandran, and P.~L. Bartlett, ``Byzantine-robust
  distributed learning: Towards optimal statistical rates,'' \emph{ArXiv}, vol.
  abs/1803.01498, 2018.

\bibitem{Fang2020LocalMP}
M.~Fang, X.~Cao, J.~Jia, and N.~Z. Gong, ``Local model poisoning attacks to
  byzantine-robust federated learning,'' \emph{ArXiv}, vol. abs/1911.11815,
  2020.

\bibitem{Tolpegin2020DataPA}
V.~Tolpegin, S.~Truex, M.~E. Gursoy, and L.~Liu, ``Data poisoning attacks
  against federated learning systems,'' in \emph{ESORICS}, 2020.

\end{thebibliography}

\end{document}